\documentclass{article}
\usepackage[a4paper]{geometry}
\usepackage{amssymb,amsmath,amsthm}
\usepackage{graphics,graphicx,color}
\usepackage{microtype,verbatim,dsfont}
\usepackage{xspace,enumerate}
\usepackage{subcaption}
\usepackage{todonotes}
\usepackage[pdfpagelabels,colorlinks,citecolor=blue,linkcolor=blue,urlcolor=blue]{hyperref}
\usepackage[capitalise]{cleveref}
\usepackage{fullpage}

\usepackage{graphicx}
\usepackage{subcaption}
\usepackage{xspace}
\usepackage{footnote,lineno}
\makesavenoteenv{tabular}
\usepackage[symbol]{footmisc}
\usepackage{booktabs}
\usepackage{enumerate,paralist}
\usepackage[basic]{complexity}
\usepackage{xcolor}
\usepackage{hyperref}
\usepackage{soul}

\theoremstyle{plain}
\newtheorem{theorem}{Theorem}

\newtheorem{corollary}[theorem]{Corollary}

\newtheorem{prp}[theorem]{Property}
\newtheorem{conj}[theorem]{Conjecture}

\Crefname{observation}{Observation}{Observations}
\Crefname{algorithm}{Algorithm}{Algorithms}
\Crefname{conjecture}{Conjecture}{Conjectures}
\Crefname{algocf}{Algorithm}{Algorithms}
\Crefname{section}{Section}{Sections}
\Crefname{lemma}{Lemma}{Lemmata}
\Crefname{lemma2}{Lemma}{Lemmata}
\Crefname{note}{Note}{Notes}
\Crefname{claim}{Claim}{Claims}
\Crefname{prp}{Property}{Properties}
\Crefname{prp2}{Property}{Properties}
\Crefname{enumi}{Condition}{Conditions}
\Crefname{conj}{Conjecture}{Conjectures}
\Crefname{figure}{Fig.}{Figs.}

\newcommand{\specialsource}{\includegraphics[page=1]{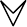}-face\xspace}
\newcommand{\specialsources}{\includegraphics[page=1]{figures/special-faces.pdf}-faces\xspace}
\newcommand{\specialsink}{\includegraphics[page=2]{figures/special-faces.pdf}-face\xspace}
\newcommand{\specialsinks}{\includegraphics[page=2]{figures/special-faces.pdf}-faces\xspace}

\newcommand{\emb}{\ensuremath{\mathcal{E}_\phi}}
\renewcommand{\emph}[1]{{\em \textcolor{blue}{#1}}\xspace}

\title{Recognizing DAGs with Page-Number 2 is NP-complete\footnote{An extended abstract of this work has appeared in the proceedings of the 30th International Symposium on Graph Drawing and Network Visualization (GD2022). This research was partially supported by $\mathrm{\Pi}$EBE 2020 and MIUR Project ``AHeAD'' under PRIN 20174LF3T8.}}

\author{Michael A. Bekos$^1$, Giordano Da~Lozzo$^2$, Fabrizio~Frati$^2$,\\Martin Gronemann$^3$, Tamara~Mchedlidze$^4$, Chrysanthi Raftopoulou$^5$
\\
\medskip
\\
\small$^1$Department of Mathematics, University of Ioannina, Ioannina, Greece\\
\small\texttt{bekos@uoi.gr}
\\
\small$^2$Department of Engineering, Roma Tre University, Italy\\
\small\texttt{giordano.dalozzo@uniroma3.it, fabrizio.frati@uniroma3.it}
\\
\small$^3$Algorithms and Complexity Group, TU Wien, Vienna, Austria\\
\small\texttt{mgronemann@ac.tuwien.ac.at}
\\
\small$^4$Department of Computer Science, Utrecht University, Utrecht, the Netherlands\\
\small\texttt{t.mtsentlintze@uu.nl}
\\
\small$^5$School of Applied Mathematical \& Physical Sciences, NTUA, Athens, Greece\\
\small\texttt{crisraft@mail.ntua.gr}
}
\date{}

\begin{document}

\maketitle

\begin{abstract}
The page-number of a directed acyclic graph (a DAG, for short) is the minimum $k$ for which the DAG has a topological order and a $k$-coloring of its edges such that no two edges of the same color cross, i.e., have alternating endpoints along the topological order. 
In 1999, Heath and Pemmaraju conjectured that the recognition of DAGs with page-number~$2$ is \NP-complete and proved that recognizing DAGs with page-number~$6$ is \NP-complete [%
{\em SIAM J. Computing}, 1999]. Binucci et al.\ recently strengthened this result by proving that recognizing DAGs with page-number~$k$ is \NP-complete, for every $k\geq 3$ [%
{\em SoCG} 2019]. In this paper, we finally resolve Heath and Pemmaraju's conjecture in the affirmative. In particular, our \mbox{\NP-completeness} result holds even for $st$-planar graphs and planar posets.


\end{abstract}

\section{Introduction}

The problem of embedding graphs in books~\cite{Oll73} has a long history of research with early results dating back to the 1970's. Such embeddings are specified by a linear order of the vertices along a line, called \emph{spine}, and by a partition of the edges into sets, called \emph{pages}, such that the edges in each page are drawn crossing-free in a half-plane delimited by the spine. The \emph{page-number} of a graph is the minimum number of pages over all its book embeddings, while the page-number of a graph family is the maximum page-number over its members.


An important branch of literature focuses on the page-number of planar graphs. An upper bound of $4$ was known since 1986~\cite{DBLP:journals/jcss/Yannakakis89}, while a matching lower bound was only recently proposed~\cite{DBLP:journals/jocg/KaufmannBKPRU20,DBLP:journals/jctb/Yannakakis20}. Better bounds are known for several families of planar graphs~\cite{DBLP:journals/dam/GuanY20,DBLP:conf/focs/Heath84}. A special attention has been devoted to the planar graphs with page-number $2$~\cite{DBLP:journals/algorithmica/BekosGR16,DBLP:journals/mp/CornuejolsNP83,Ewald1973,DBLP:journals/dcg/FraysseixMP95,DBLP:conf/esa/0001K19,DBLP:journals/appml/KainenO07,NC08,DBLP:conf/cocoon/RengarajanM95}. These have been characterized as the subgraphs of the Hamiltonian planar graphs~\cite{doi:10.1137/0608018} and hence are called \emph{subhamiltonian}. Recognizing subhamiltonian graphs turns out to be NP-complete~\cite{Wig82}.


If the input graph is directed and acyclic (a DAG, for short), then the linear vertex order of a book embedding is required to be a \emph{topological order} of it~\cite{DBLP:journals/order/NowakowskiP89}. Heath and Pemmaraju~\cite{DBLP:journals/siamcomp/HeathP99} showed that there exist \emph{planar DAGs} (i.e., DAGs whose underlying graph is planar) whose page-number is linear in the input size. Certain subfamilies of planar DAGs, however, have bounded page-number~\cite{DBLP:conf/gd/AlzohairiR96,DBLP:conf/gd/BhoreLMN21,DBLP:journals/algorithmica/GiacomoDLW06,DBLP:journals/siamcomp/HeathPT99}. Further, it was recently shown that \emph{upward planar graphs} (i.e., DAGs that admit an upward planar drawing, where \emph{upward} means that each edge is represented by a curve whose $y$-coordinates monotonically increase from the source to the sink of the edge)
have sublinear page-number~\cite{DBLP:conf/soda/JungeblutMU22}, improving upon previous bounds~\cite{DBLP:journals/jgaa/FratiFR13}. From an algorithmic point of view, testing whether~a~DAG has page-number $k$ is \NP-complete for every fixed value of $k \geq 3$~\cite{binucci_et_al2019}, linear-time solvable for $k=1$~\cite{DBLP:journals/siamcomp/HeathP99}, and fixed-parameter tractable with respect to the vertex cover number for every $k$~\cite{DBLP:conf/gd/BhoreLMN21} and with respect to the treewidth for \emph{$st$-graphs} (i.e., DAGs with a single source and a single sink) when $k=2$~\cite{binucci_et_al2019}. In contrast to the undirected setting, however, for $k=2$ the complexity question has remained open since 1999, when Heath and Pemmaraju posed the following conjecture.

\begin{conj}[Heath and Pemmaraju~\cite{DBLP:journals/siamcomp/HeathP99}]\label{conj:main}
Deciding whether a DAG has page-number 2 is \NP-complete. 
\end{conj}

\noindent\textbf{Our contribution.} In this work, we settle~\cref{conj:main} in the affirmative. More precisely, in \cref{sec:st-complete}, we show that testing whether an \emph{$st$-planar graph} (i.e., an $st$-graph that admits a planar drawing in which the source and the sink are incident to the outer face) admits a $2$-page book embedding is an \NP-complete problem. In \cref{se:poset}, we further show that the problem is \NP-complete also for \emph{planar posets} (i.e., upward planar graphs with no transitive edges). \cref{sec:preliminaries} contains some definitions and preliminaries, while \cref{sec:conclusions} presents some conclusions and open problems.

\section{Preliminaries} \label{sec:preliminaries}
%


A \emph{combinatorial embedding} of a graph is an equivalence class of planar drawings of the graph, where two drawings are equivalent if they define the same clockwise order of the incident edges at each vertex. A \emph{plane embedding} of a connected graph is an equivalence class of planar drawings of the graph, where two drawings are equivalent if they define the same combinatorial embedding and the same clockwise order of the vertices along the outer face. The \emph{flip} of a plane embedding produces a plane embedding in which the clockwise order of the incident edges at each vertex and the clockwise order of the vertices along the outer face is the reverse of the original one. An \emph{upward planar embedding} is an equivalence class of upward planar drawings of a DAG, where two drawings are equivalent if they define the same plane embedding and the same left-to-right order of the outgoing (and incoming) edges at each vertex. A \emph{plane DAG} is a DAG together with an upward planar embedding. 
It is known that every $st$-planar graph is upward planar~\cite{DBLP:journals/tcs/BattistaT88,DBLP:journals/dm/Kelly87}. 
An \emph{$st$-plane} graph is an $st$-planar graph together with an upward planar embedding in which $s$ and $t$ are incident to the outer face. 

As in the undirected case, a DAG $G$ has page-number $2$ if it is \emph{subhamiltonian}, i.e., it is a spanning subgraph of an $st$-planar graph $\overline{G}$ that has a directed Hamiltonian $st$-path $P$~\cite{DBLP:journals/jgaa/MchedlidzeS11}. In the previous definition, if $G$ has a prescribed plane embedding, we additionally require that the
plane embedding of $\overline{G}$ restricted to $G$ coincides with the one of $G$. We say that $P$ is a \emph{subhamiltonian path} for $G$, and we refer to the edges of $P$ that are not in~$G$ as  \emph{augmenting edges}. Further, $\overline{G}$ is called an \emph{HP-completion} of $G$.

A \emph{generalized diamond} is an $st$-plane graph consisting of three directed paths from a vertex $v_s$ to a vertex $v_t$, one of which is the edge $v_s v_t$ and appears between the other two paths in the upward planar embedding; see \cref{fig:diamonds1}.

\begin{figure}[t]
	\centering
	\begin{subfigure}[b]{.08\textwidth}
		\centering
		\includegraphics[page=6]{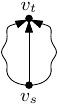}
		\caption{}
		\label{fig:diamonds1}
	\end{subfigure}
	\hfil
	\begin{subfigure}[b]{.08\textwidth}
		\centering
		\includegraphics[page=7]{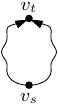}
		\caption{}
		\label{fig:diamonds2}
	\end{subfigure}
	\hfil
	\begin{subfigure}[b]{.1\textwidth}
		\centering
		\includegraphics[page=8]{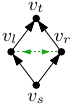}
		\caption{}
		\label{fig:diamonds3}
	\end{subfigure}
	\caption{(a)~A generalized diamond, (b)~a non-transitive face, and (c)~a rhombus; curly curves represent paths and straight-line segments represent edges.}
	\label{fig:diamond-nontransitive}
\end{figure}

Unless otherwise specified, by {\em face} of a plane DAG we always mean an {\em internal face}. A face of a plane DAG whose boundary consists of two directed paths is an \emph{$st$-face}. An $st$-face is \emph{transitive} if one of these paths is an edge; otherwise, it is \emph{non-transitive} (see \cref{fig:diamonds2}).
A \emph{rhombus} is a non-transitive $st$-face whose boundary paths have  length $2$; see \cref{fig:diamonds3}.
The following property follows from Theorem 1 in  \cite{DBLP:journals/jgaa/MchedlidzeS11}.

\begin{prp}\label{prop:generalized_diamond}
A Hamiltonian $st$-plane graph contains only transitive faces and no generalized diamond. 
\end{prp}

\noindent From the above property we can deduce the following.

\begin{prp}\label{prop:rhombus_consec}
Let $G$ be a plane DAG and $P$ be a subhamiltonian path for $G$. If~$G$ contains a rhombus $(v_s, v_l, v_r, v_t)$ with source $v_s$ and sink $v_t$, then $P$ contains either the edge $v_l v_r$ or the edge $v_r v_l$, i.e., $v_l$ and $v_r$ are consecutive in $P$.
\end{prp}

\noindent The next property follows directly from Theorem~1 in~\cite{MchedlidzeS09} and Property~\ref{prop:generalized_diamond}.

\begin{prp}\label{prp:non_transitive_augmenting} 
Let $G$ be a plane DAG and $P$ be a subhamiltonian path for $G$. If $G$ contains a non-transitive face $f$ with boundaries $(v_s, w, v_t)$ and $(v_s, v_1,\dots,v_r,\allowbreak v_t)$, then the augmenting edges of $P$ inside $f$ are either (i) the edge $w v_1$, or (ii) the edge $v_r w$, or (iii) edges $v_iw$ and $wv_{i+1}$ for some $1\leq i<r$. 
\end{prp} 
\begin{proof}
Consider any HP-completion $\overline{G}$ of $G$ with subhamiltonian path $P$. By \cref{prop:generalized_diamond}, we have that $\overline{G}$ does not contain any non-transitive face and any generalized diamond. Since $f$ is non-transitive and $\overline{G}$ cannot contain the edge $v_sv_t$ inside $f$, as this would create a generalized diamond with the boundary paths of $f$, it follows that $P$ has at least one augmenting edge inside $f$ that connects $w$ and a vertex of its right boundary. Since $w$ can be incident to at most two edges of $P$, there can be at most two augmenting edges of $P$ inside~$f$. 

\begin{figure}[t]
	\centering
	\begin{subfigure}[b]{.24\textwidth}
		\centering
		\includegraphics[page=1]{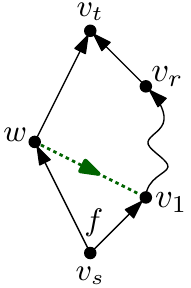}
		\caption{Case (i)}
		\label{fig:face_augmenting_1_edge_a}
	\end{subfigure}
	\hfil
	\begin{subfigure}[b]{.24\textwidth}
		\centering
		\includegraphics[page=2]{figures/faces_augmenting_edges.pdf}
		\caption{Case (ii)}
		\label{fig:face_augmenting_1_edge_b}
	\end{subfigure}
	\hfil
	\begin{subfigure}[b]{.24\textwidth}
		\centering
		\includegraphics[page=3]{figures/faces_augmenting_edges.pdf}
		\caption{Case (iii)}
		\label{fig:face_augmenting_2_edges}
	\end{subfigure}
	\caption{Illustrations for \cref{prp:non_transitive_augmenting}.}
	\label{fig:non_transitive_augmenting}
\end{figure}

If there is only one such edge, again by \cref{prop:generalized_diamond}, this edge must split $f$ into two transitive faces. This is achieved only by the edges $w v_1$ and $v_r w$, which implies cases (i) and (ii) of the statement; see \cref{fig:face_augmenting_1_edge_a,fig:face_augmenting_1_edge_b}. 
On the other hand, if there are two such edges, say $v_i w$ and $w v_j$ with $1\leq i<j\leq r$, then $j=i+1$ holds (refer to \cref{fig:face_augmenting_2_edges}), as otherwise $\overline{G}$ would contain a non-transitive face with left boundary $(v_i, w, v_j)$ and right boundary $(v_i,v_{i+1},\dots,v_j)$, contradicting \cref{prop:generalized_diamond}. Hence case (iii) holds.
\end{proof}

\section{\NP-completeness Proof for Planar $st$-Graphs} \label{sec:st-complete}

\begin{figure}[b]
    \centering
    \includegraphics[page=7]{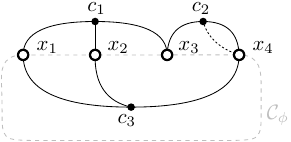}
    \caption{The incidence graph $G_\phi$ of an instance $\phi$ of {\sc Planar Monotone 3-SAT} with $\phi = c_1 \land c_2 \land c_3$, $c_1=(x_1 \lor x_2 \lor x_3)$, $c_2=(x_3 \lor x_4)$, and $c_3=(\overline{x}_1 \lor \overline{ x}_2 \lor \overline{x}_4)$.}
    \label{fig:pm3sat}
\end{figure}

Let $\phi$ be a Boolean $3$-SAT formula with $n$ variables $x_1,\ldots,x_n$ and $m$ clauses $c_1,\dots,c_m$. 
A clause of $\phi$ is \emph{positive} (\emph{negative}) if it has only positive (negative) literals. 
The \emph{incidence graph} $G_\phi$ of $\phi$ is the graph that has \emph{variable vertices} $x_1,\ldots,x_n$, \emph{clause vertices} $c_1,\dots,c_m$, and has an edge $(c_j,x_i)$ for each clause $c_j$ containing $x_i$ or $\overline{x}_i$. 
Note that we use the same notation for variables (clauses) in $\phi$ and variable vertices (clause vertices) in $G_\phi$.
If $\phi$ has clauses with less~than three literals, we introduce parallel edges in $G_\phi$ so that all clause vertices have degree~$3$ in $G_\phi$; see, e.g., the dotted edge $(c_2,x_4)$ in \cref{fig:pm3sat}.
%
The formula~$\phi$ is an instance of the \NP-complete {\sc Planar Monotone 3-SAT} problem~\cite{DBLP:journals/ijcga/BergK12} if each clause of $\phi$ is positive or negative, and $G_\phi$ has a plane embedding~$\emb$ to which the edges of a cycle $\mathcal{C}_\phi:= x_1,\ldots,x_n$ can be added 
that separates positive and negative clause vertices. The problem asks whether~$\phi$ is satisfiable. 

Next, we present~gadgets that we are going to use in a reduction from {\sc Planar Monotone 3-SAT} to the problem of deciding whether a given $st$-planar graph admits a $2$-page book embedding.

\paragraph{Double ladder.} A double ladder of even length $\ell$ is defined as follows; see \cref{fig:ladder}. 
Its vertex set consists of two sources, $s_1$ and $s_2$, two sinks, $t_1$ and $t_2$, and vertices in $\cup_{i=0}^{\ell}\{u_i,v_i,w_i\}$. Its edge set consists of edges $s_1u_0$, $s_1v_0$, $s_2v_0$, $s_2w_0$, $u_{\ell}t_1$, $v_{\ell}t_1$, $v_{\ell}t_2$, $w_{\ell}t_2$, and $\cup_{i=0}^{\ell-1}\{u_iu_{i+1}, v_iu_{i+1}, v_iv_{i+1}, w_iv_{i+1}, w_iw_{i+1}\}$.

\begin{figure}[t]
	\centering
	\begin{subfigure}[b]{.32\textwidth}
		\centering
		\includegraphics[page=6]{figures/variable_gadget.pdf}
		\caption{A double ladder of length $\ell$}
		\label{fig:ladder}
	\end{subfigure}
	\begin{subfigure}[b]{.32\textwidth}
		\centering
		\includegraphics[page=1]{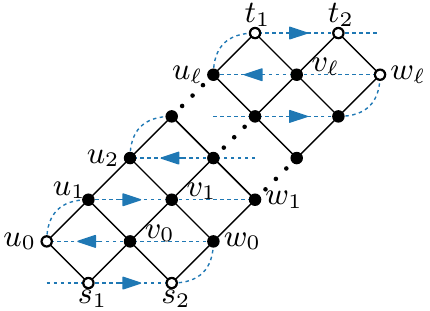}
		\caption{\texttt{true} assignment}
		\label{fig:variable_gadget_true}
	\end{subfigure}
	\begin{subfigure}[b]{.32\textwidth}
		\centering
		\includegraphics[page=2]{figures/variable_gadget.pdf}
		\caption{\texttt{false} assignment}
		\label{fig:variable_gadget_false}
	\end{subfigure}
	\caption{(a)~A double ladder of length $\ell$, and (b)-(c)~the two subhamiltonian paths of it; edges with no arrow are directed upward, also in subsequent figures.}
	\label{fig:variable_gadget}
\end{figure}

\begin{prp}\label{prp:embedded-ladder}
The double ladder has a unique upward planar embedding (up to a flip),  shown in \cref{fig:ladder}.
\end{prp}
\begin{proof}
The embedding shown in \cref{fig:ladder} clearly is an upward planar embedding of the double ladder. The underlying graph of the double ladder has four combinatorial embeddings, which are obtained from the embedding in \cref{fig:ladder}, by possibly moving the path $u_1u_0s_1v_0$ inside the cycle $u_1u_2v_1v_0$ and the path $w_{\ell-1}w_\ell t_2v_\ell$ inside the cycle $w_{\ell-1}w_{\ell-2} v_{\ell-1}v_\ell$. However, these movements respectively force $s_1v_0$ and $v_\ell t_2$ to point downward, hence the resulting combinatorial embeddings do not correspond to upward planar embeddings. Finally, since the outer face of the embedding in \cref{fig:ladder} is the only face containing at least one source and one sink of the double ladder, the claim follows.
\end{proof}

\begin{prp}
Let $G$ be a plane DAG with a subhamiltonian path $P$. If $G$ contains a double ladder of length $\ell$, then $P$ contains the pattern $[\ldots u_i v_i w_i \ldots w_{i+1}\allowbreak v_{i+1} u_{i+1} \ldots]$ or 
$[\ldots w_i v_i u_i\ldots u_{i+1} v_{i+1} w_{i+1} \ldots]$ for  $i=0,\ldots,\ell-1$.
\end{prp}
\begin{proof}
By \cref{prop:rhombus_consec,prp:embedded-ladder}, we have that $P$ contains either the subpath $u_iv_iw_i$ or the subpath $w_iv_iu_i$, for  $i=0,\ldots,\ell$. The edge $u_iu_{i+1}$ then implies that the vertices $u_i,v_i,w_i$ precede the vertices $u_{i+1},v_{i+1},w_{i+1}$ in $P$. So, it remains to~rule out  patterns $[\ldots u_i v_i w_i \ldots u_{i+1} v_{i+1} \allowbreak w_{i+1}\ldots]$ and $[\ldots w_i v_i u_i \ldots w_{i+1}\allowbreak v_{i+1} \allowbreak u_{i+1}\ldots]$. If $P$ contains one of them, then any book embedding of $G$ in which the order of the vertices along the spine is the one in $P$ requires at least $3$ pages, given that the edges $u_i u_{i+1}$, $v_i v_{i+1}$ and $w_i w_{i+1}$ pairwise cross. This contradicts the fact that $P$ is a subhamiltonian path for $G$. 
\end{proof}

\begin{corollary}\label{cor:double-ladder}
There exist two subhamiltonian paths for the double ladder, shown in \cref{fig:variable_gadget_true,fig:variable_gadget_false}. 
\end{corollary}

\paragraph{Variable gadget:} Let $x\in\{x_1,\ldots,x_n\}$. The variable gadget $L_x$ for $x$ is the double ladder of length $4d_x$, where $d_x$ is the degree of $x$ in $G_\phi$. To distinguish between vertices of different variable gadgets, we denote the vertices of $L_x$ with the superscript $x$, as in \cref{fig:connector_gadget}. 
Vertices $s^x_1$, $s^x_2$, $u^x_0$ are the \emph{bottom connectors} and $w^x_{4d_x}$, $t^x_1$,  $t^x_2$ are the \emph{top connectors} of $L_x$.
The two subhamiltonian paths of \cref{cor:double-ladder} correspond to the truth assignments of $x$;  \cref{fig:variable_gadget_true} corresponds to $\texttt{true}$, while  \cref{fig:variable_gadget_false} to $\texttt{false}$. Also, we refer to the edges of $L_x$ that are part of the subhamiltonian path of  \cref{fig:variable_gadget_true} (of \cref{fig:variable_gadget_false}) as \emph{true edges} (\emph{false edges}, respectively). 
In particular, $u^x_{2j}u^x_{2j+1}$ and  $w^x_{2j+1}w^x_{2j+2}$ are true edges of $L_x$, while $u^x_{2j+1}u^x_{2j+2}$ and $w^x_{2j}w^x_{2j+1}$ are false edges of $L_x$, for $j=0,\dots,2d_x-1$. 

\paragraph{Connector gadget:} A connector gadget \emph{joins} two variable gadgets $L_x$ and $L_y$ by means of three paths from the top connectors of $L_x$ to the bottom connectors of $L_y$; see \cref{fig:connector_gadget}. These paths are: the edge $t^x_1u^y_0$, the length-$2$ path $t^x_2 \rho_{x,y} s^y_1$, where $\rho_{x,y}$ is a newly introduced vertex, and the edge $w^x_{4d_x}s^y_2$. 

\begin{figure}[t]
\begin{subfigure}[b]{.48\textwidth}
    \centering
    \includegraphics[page=1]{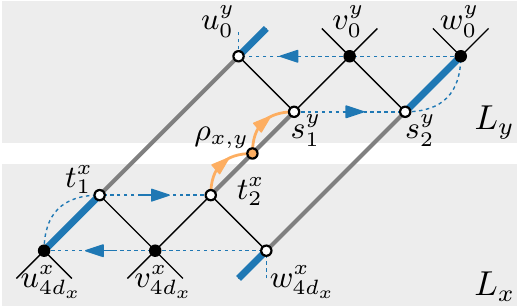}
    \caption{$L_x \leftarrow \texttt{true}$; $L_y \leftarrow \texttt{true}$}
    \label{fig:connector_gadget_TT}
\end{subfigure}
\hfil
\begin{subfigure}[b]{.48\textwidth}
    \centering
    \includegraphics[page=3]{figures/connector_gadget.pdf}
    \caption{$L_x \leftarrow \texttt{false}$; $L_y \leftarrow \texttt{false}$}
    \label{fig:connector_gadget_FF}
\end{subfigure}

\medskip

\begin{subfigure}[b]{.48\textwidth}
    \centering
    \includegraphics[page=2]{figures/connector_gadget.pdf}
    \caption{$L_x \leftarrow \texttt{true}$; $L_y \leftarrow \texttt{false}$}
    \label{fig:connector_gadget_TF}
\end{subfigure}
\hfil
\begin{subfigure}[b]{.48\textwidth}
    \centering
    \includegraphics[page=4]{figures/connector_gadget.pdf}
    \caption{$L_x \leftarrow \texttt{false}$; $L_y \leftarrow \texttt{true}$}
    \label{fig:connector_gadget_FT}
\end{subfigure}
\caption{A subhamiltonian path for the graph composed of two variable gadgets joined by a connector gadget, if the subhamiltonian paths for the variable gadgets represent (a)-(b) the same truth assignment, and (c)-(d) the opposite truth assignment.}\label{fig:connector_gadget}
\end{figure}


\begin{prp}\label{prop:connector}
Given subhamiltonian paths $P_x$ for $L_x$ and $P_y$ for $L_y$, there is a subhamiltonian path $P$ containing $P_x$ and $P_y$ for the graph obtained by joining $L_x$ and $L_y$ by means of a connector gadget.
\end{prp}
\begin{proof}
Each of $P_x$ and $P_y$ is one of the two subhamiltonian paths of \cref{cor:double-ladder}; see \cref{fig:variable_gadget}. In particular, the last vertex of $P_x$ is $t_1^x$ or $t_2^x$, and the first vertex of $P_y$ is $s_1^y$ or $s_2^y$. We obtain $P$ by adding a directed edge from the last vertex of $P_x$ to $\rho_{x,y}$ and a direct edge from $\rho_{x,y}$ to the first vertex of $P_y$, as illustrated in \cref{fig:connector_gadget}.
\end{proof}

\paragraph{Clause gadget:} Let $c$ be a positive (negative) clause of $\phi$. Assume that the variables $x$, $y$ and $z$ of $c$ appear in this order along $\mathcal{C}_\phi$, when traversing $\mathcal{C}_\phi$ from $x_1$ towards $x_n$. In $\emb$, the edges between $x$ and the positive (negative) clause vertices of $G_\phi$ appear consecutively around $x$. Assume that the edge $(c,x)$ is the $(i+1)$-th such edge in a clockwise (counter-clockwise) traversal of the edges around $x$ starting at the edge of $\mathcal{C}_\phi$ incoming $x$; note that $i\in \{0,\dots,d_x-1\}$. Similarly, define indices $j$ and $k$ for $y$ and $z$, respectively. For convenience, we refer to $i$, $j$ and $k$ as the \emph{clause-indices} of $x$, $y$ and $z$, respectively. Let $L_x$, $L_y$, and $L_z$ be the three variable gadgets for $x$, $y$, and $z$, respectively. 

\begin{figure}[htb]
	\centering
	\begin{subfigure}[b]{.48\textwidth}
		\includegraphics[page=1]{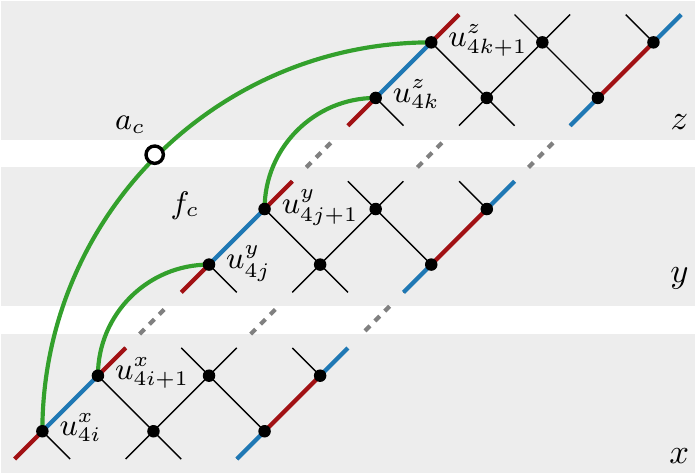}
		\caption{$(x \lor y \lor z)$}
		\label{fig:positive_clause_gadget}
	\end{subfigure}
	\hfil
	\begin{subfigure}[b]{.48\textwidth}
		\includegraphics[page=2]{figures/clause_gadget.pdf}
		\caption{$(\overline{x} \lor \overline{y} \lor \overline{z})$}
		\label{fig:negative_clause_gadget}
	\end{subfigure}
	\caption{Clause gadgets for (a)~a positive clause and (b)~a negative clause.}
	\label{fig:clause_gadget}
\end{figure}

The clause gadget $C_c$ for $c$ consists of an \emph{anchor vertex}~$a_c$, and four edges. If $c$ is positive, these edges are $u^x_{4i}a_c$, $a_cu^z_{4k+1}$, $u^x_{4i+1}u^y_{4j}$ and $u^y_{4j+1}u^z_{4k}$ (green~in \cref{fig:positive_clause_gadget}); otherwise, they are $w^x_{4i}a_c$, $a_cw^z_{4k+1}$, $w^x_{4i+1}w^y_{4j}$ and $w^y_{4j+1}w^z_{4k}$ (green~in \cref{fig:negative_clause_gadget}).  Note that $C_c$ creates a non-transitive face $f_c$, called \emph{anchor face}, whose boundary is delimited by the two newly-introduced edges incident to $a_c$ and by a directed path whose edges alternate between three true edges (if $c$ is positive) or three false edges (if $c$ is negative) and the two newly-introduced edges not incident to $a_c$. We refer to the three true (or false) edges on the boundary of $f_c$ stemming from $L_x$, $L_y$, and $L_z$ as the \emph{base-edges} of $C_c$.
The length of the double ladders ensures that, if $x=y$ (which implies that $j=i+1$), then vertices $u^x_{4i+1}$ and $u^y_{4j}$ ($w^x_{4i+1}$ and $w^y_{4j+1}$) are not adjacent in $L_x$ and the edge $u^x_{4i+1}u^y_{4j}$ ($w^x_{4i+1}w^y_{4j+1}$) is well defined; this is the reason that we do not use vertices with indices $2,3 \mod 4$.

\medskip\noindent We are now ready to prove the main theorem of this section.
\newpage

\begin{theorem}\label{thm:st}
Recognizing whether a DAG has page-number~$2$ is \NP-complete, even if the input is an $st$-planar graph.
\end{theorem}
\begin{proof}
The problem clearly belongs to \NP, as a non-deterministic Turing machine can guess an order of the vertices of an input graph and a partition of its edges into two pages, and check in polynomial time whether the order is a topological order and if so, whether any two edges in the same page cross.

Given an instance 
$\phi$ of {\sc Planar Monotone $3$-SAT}, we construct in polynomial time an $st$-planar graph $H$ that has page-number $2$ if and only if $\phi$ \mbox{is satisfiable; see \cref{fig:construction_example}.} We consider the variable gadgets $L_{x_1}$, $\dots$, $L_{x_n}$, where $x_1, \dots, x_n$ is the order of the variables of $\phi$ along the cycle $\mathcal{C}_\phi$; for $i=1,\dots,n-1$, we connect $L_{x_i}$ with $L_{x_{i+1}}$ using a connector gadget. 
For each positive (negative) clause $c$ of $\phi$, we add a clause gadget $C_c$. 
By the choice of the clause-indices, we can deduce that the resulting graph is a plane DAG containing two sources $s^{x_1}_1$ and $s^{x_1}_2$ and two sinks $t^{x_n}_1$ and $t^{x_n}_2$. We add a source $s$ connected with outgoing edges to $s^{x_1}_1$ and $s^{x_1}_2$, and a sink $t$ connected with incoming edges from $t^{x_n}_1$ and $t^{x_n}_2$. The constructed graph $H$ is $st$-planar. Since  the underlying graph of $H$ is a subdivision of a triconnected planar graph and since only one face of $H$ contains $s$ and $t$, it follows that $H$ has a unique upward planar embedding. We next prove that $\phi$ is satisfiable if and only if $H$ is subhamiltonian (and therefore has page-number $2$).

For the forward direction, assuming that $\phi$ admits a satisfying truth assignment, we show how to construct a subhamiltonian path $P$ for $H$.
For $i=1,\dots,n$, we have that $P$ contains the subhamiltonian path $P_i$ for $L_{x_i}$ shown in \cref{fig:variable_gadget_true} if $x_i$ is \texttt{true}, and the one shown in \cref{fig:variable_gadget_false} otherwise. By \cref{prop:connector}, there is a subhamiltonian path $P$ for the subgraph of $H$ induced by the vertices of all variable and connector gadgets, containing $P_1,\dots,P_n$ as subpaths.
The path $P$ starts from a source of $L_{x_1}$ and ends at a sink of $L_{x_n}$; hence we can extend $P$ to include $s$ and $t$ as its first and last vertices, respectively. 
We now extend $P$ to a subhamiltonian path for $H$ by including the anchor vertex of each clause gadget.
Consider a positive clause $c=(x\lor y\lor z)$ with anchor vertex $a_c$; the case of a negative clause is analogous. As $\phi$ is satisfied, at least one of $x$, $y$ and $z$ is \texttt{true}; assume that $x$ is \texttt{true}, as the other two cases are analogous. By construction, the anchor face $f_c$~of~$C_c$ is non-transitive, with the anchor vertex $a_c$ on its left boundary, and the three base-edges of $C_c$ along its right boundary.   
Recall that each of these base-edges belong to $L_x$, $L_y$, and $L_z$, respectively.
Let $i\in \{0,\dots,d_x-1\}$ be such that
$u^x_{4i}u^x_{4i+1}$ is the (true) base-edge of $L_x$ on the right boundary of $f_c$. Since $x$ is \texttt{true}, the vertices $u^x_{4i}$ and $u^x_{4i+1}$ are consecutive in $P$. We extend $P$ by visiting vertex $a_c$ after $u^x_{4i}$ and before $u^x_{4i+1}$. This corresponds to adding two augmenting edges $u^x_{4i}a_c$ and $a_cu^x_{4i+1}$ of $P$ in the interior of $f_c$; see the black dashed edges of \cref{fig:construction_example}. 
At the end of this process, $P$ is extended to a subhamiltonian path for~$H$. 

\begin{figure}[t!]
\centering
\includegraphics[page=6, width=\textwidth]{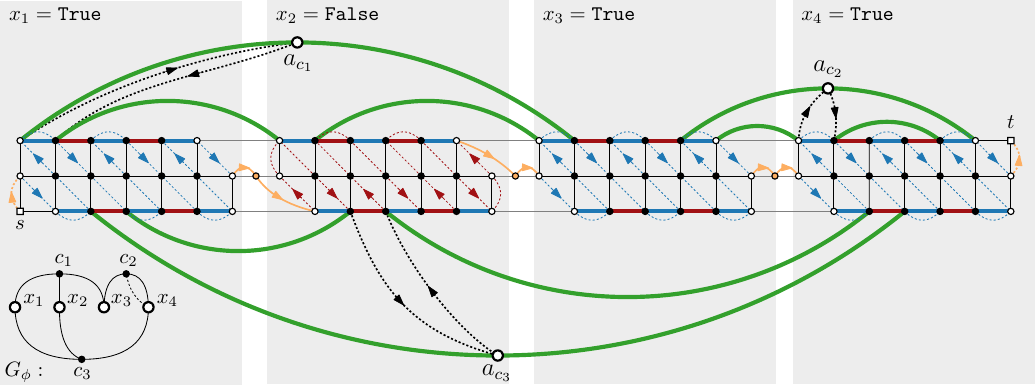}
\caption{The graph $H$ obtained from the instance of {\sc Planar Monotone $3$-SAT} illustrated in \cref{fig:pm3sat}. 
For space reasons, the double ladders of 
the variable gadgets have smaller length and the drawing is rotated by $45^\circ$.}\label{fig:construction_example}
\end{figure}

For the other direction, assume that $H$ is subhamiltonian and let $P$ be  subhamiltonian path for it. 
For each variable gadget $L_{x_i}$, $P$ induces a subhamiltonian path $P_{i}$ for $L_{x_i}$. By \cref{cor:double-ladder}, $P_i$ is one of the two subhamiltonian paths of \cref{fig:variable_gadget}. We assign to $x_i$ the value \texttt{true} if $P_i$ is the path of \cref{fig:variable_gadget_true} and \texttt{false} if $P_i$ is the path of  \cref{fig:variable_gadget_false}. We claim that this truth assignment satisfies $\phi$. Assume, for a contradiction, that there exists a clause $c$ that is not satisfied. Also assume that $c$ is a positive clause  $(x\lor y\lor z)$ (the case of $c$ being negative is analogous). Without loss of generality, we can further assume that $x$, $y$ and $z$ appear in this order in $C_{\phi}$, 
and that the base-edges of the clause gadget $C_c$ along the right boundary of the anchor face $f_c$ are the true edges $u^x_{4i}u^x_{4i+1}$, $u^y_{4j}u^y_{4j+1}$, and $u^z_{4k}u^z_{4k+1}$ of $L_x$, $L_y$, and $L_z$, respectively. 
Since clause $c$ that is not satisfied, $x$, $y$ and $z$ are \texttt{false}, which implies that the corresponding subhamiltonian paths $P_x$, $P_y$ and $P_z$ of $L_x$, $L_y$ and $L_z$ are the ones of \cref{fig:variable_gadget_false}. Hence, $P$ contains the augmenting edges $u^x_{4i}v^x_{4i}$ and $v^x_{4i+1}u^x_{4i+1}$ of $P_x$, $u^y_{4j}v^y_{4j}$ and $v^y_{4j+1}u^y_{4j+1}$ of $P_y$ and  $u^z_{4k}v^z_{4k}$ and $v^z_{4k+1}u^z_{4k+1}$ of $P_z$. 
By \cref{prp:non_transitive_augmenting} for the non-transitive face $f_c$, $P$ contains  either (i) the augmenting edge $a_c u^x_{4i+1}$,  or (ii) the augmenting edge $u^z_{4k}a_c$, or (iii) for a pair of consecutive vertices, say $u$ and $u'$, along the right boundary of $f_c$, the augmenting edges $ua_c$ and $a_cu'$. Cases (i) and (ii) contradict the existence of augmenting edges $v^x_{4i+1}u^x_{4i+1}$ and  $u^z_{4k}v^z_{4k}$ of $P$, respectively. Further, in case (iii) the augmenting edges of $P$ that belong to $P_x$, $P_y$, and $P_z$ imply that $u\notin\{ u^x_{4i}, u^y_{4j}, u^z_{4k}\}$ and $u'\notin\{ u^x_{4i+1}, u^y_{4j+1}, u^z_{4k+1}\}$. Hence $u=u^x_{4i+1}$ and $u'=u^y_{4j}$ holds, or $u=u^y_{4j+1}$ and $u'=u^z_{4k}$ holds. In both cases, the HP-completion of $H$ contains a generalized diamond with $v_s=u$ and $v_t=u'$ (the two directed paths on the sides of the edge $v_sv_t$ are $v_sa_cv_t$ and the path composed of the edges of the ladder and connector gadgets from $v_s$ to $v_t$), violating \cref{prop:generalized_diamond}.
%
%
Hence at least one of variables $x$, $y$ and $z$ must be \texttt{true}, contradicting our assumption that $c$~is~not~satisfied.
\end{proof}

\section{NP-completeness Proof for Planar Posets}\label{se:poset}
In this section, we show that the problem of determining whether a DAG has page-number~$2$ is \NP-complete also if the input graph is a planar poset. To show this, we increase the length of the double ladders used for the variable gadgets and modify the clause gadget, while keeping the same connector gadget. In particular for a variable $x \in \{x_1,\ldots,x_n\}$, the variable gadget $L_{x}$ is the double ladder of length $6d_x$, where $d_x$ is the degree of $x$ in the incidence graph $G_\phi$. 

For a clause $c$ of $\phi$ whose variables $x$, $y$ and $z$ appear in this order along $\mathcal{C}_\phi$ starting from $x_1$ towards $x_n$, define the clause-indices $i$, $j$ and $k$ of $c$ in the exact same way as in \cref{sec:st-complete}. Then, the clause gadget $C_c$ corresponding to $c$ consists of $11$ vertices ($a_c$, $a_c^1$, $a_c^2$, $s_c$, $s_c^1$, $s_c^2$, $t_c$, $t_c^1$, $t_c^2$, $v_c^1$ and $v_c^2$) and $21$ edges defined as follows. If $c$ is positive, these edges are 
$u^x_{6i+1}a_c$, $a_cu^z_{6k+4}$, $s_ca_c$, $s_cv_c^1$, $a_ct_c$, $v_c^2t_c$, $v_c^1v_c^2$, $v_c^2u^z_{6k+3}$,
$u^x_{6i+3}a_c^1$, $a_c^1u^y_{6j+2}$, $s_c^1a_c^1$, $s_c^1 u^x_{6i+5}$, $a_c^1t_c^1$, $u^y_{6j}t_c^1$, 
$u^y_{6j+3}a_c^2$, $a_c^2u^y_{6k+2}$, $s_c^2a_c^2$, $s_c^2u^y_{6i+5}$, $a_c^2t_c^2$, $u^z_{6k}t_c^2$ 
(green~in \cref{fig:clause_gadget_poset}); otherwise, they are 
$w^x_{6i+1}a_c$, $a_cw^z_{6k+4}$, $s_ca_c$, $s_cv_c^1$, $a_ct_c$, $v_c^2t_c$, $v_c^1v_c^2$, $v_c^2w^z_{6k+3}$,
$w^x_{6i+3}a_c^1$, $a_c^1w^y_{6j+2}$, $s_c^1a_c^1$, $s_c^1 w^x_{6i+5}$, $a_c^1t_c^1$, $w^y_{6j}t_c^1$, 
$w^y_{6j+3}a_c^2$, $a_c^2w^y_{6k+2}$, $s_c^2a_c^2$, $s_c^2w^y_{6i+5}$, $a_c^2t_c^2$, $w^z_{6k}t_c^2$.
By construction, $s_c$, $s_c^1$ and $s_c^2$ are sources, while $t_c$, $t_c^1$ and $t_c^2$ are sinks. 
The remaining vertices of the clause gadget ensure the absence of transitive edges, as required in the construction.

\begin{figure}[htb]
\centering
\includegraphics[page=4]{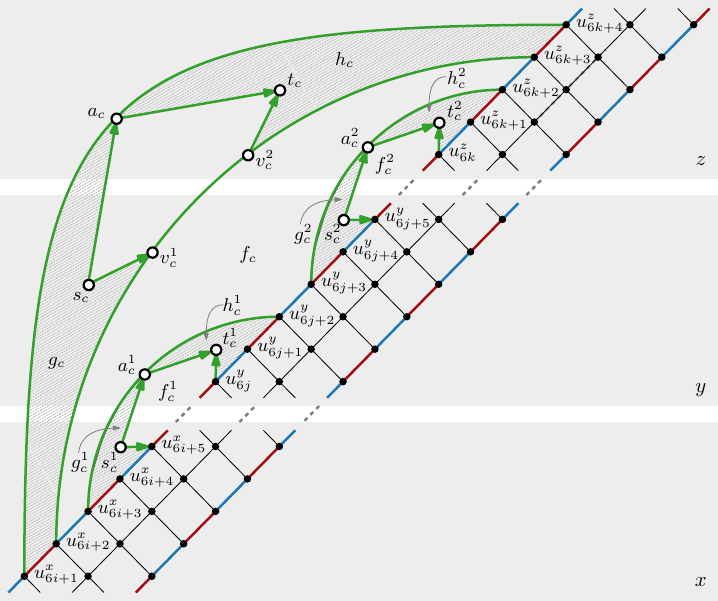}
\caption{Clause gadget with sources $s_c$, $s_c^1$ and $s_c^2$,  sinks $t_c$, $t_c^1$ and $t_c^2$, and without transitive edges. }
\label{fig:clause_gadget_poset}
\end{figure}

\begin{theorem}\label{thm:poset}
Recognizing whether a DAG has page-number~$2$ is \NP-complete, even if the input is a planar poset.
\end{theorem}
\begin{proof}
Given an instance $\phi$ of {\sc Planar Monotone $3$-SAT}, we construct a plane DAG $H$  similarly as in the proof of \cref{thm:st}. The graph $H$ is upward planar, it has multiple sources and sinks, and it does not contain any transitive edges, i.e., it is a planar poset. In particular, the absence of transitive edges and the presence of multiple sources and sinks derive from the design of the clause gadget, as we have already mentioned. In the following, we focus on proving that $\phi$ is satisfiable if and only if $H$ is subhamiltonian.

\begin{figure}[b!]
    \centering
    \begin{subfigure}[b]{0.49\textwidth}
    \centering
    \includegraphics[page=5,scale=1.2]{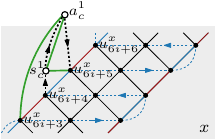}
    \caption{$x=\texttt{true}$}
    \label{fig:clause_gadget_poset_traversal_true}
    \end{subfigure}
    \hfil
    \begin{subfigure}[b]{0.49\textwidth}
    \centering
    \includegraphics[page=6,scale=1.2]{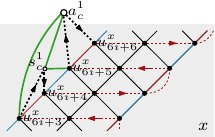}
    \caption{$x=\texttt{false}$}
    \label{fig:clause_gadget_poset_traversal_false}
    \end{subfigure}
    \caption{Extending subhamiltonian path $P$ to include $a_c^1$ and $s_c^1$ in the case in which $x$ is (a)~$\texttt{true}$, and (b)~$\texttt{false}$.}
\end{figure}

First, we show how to construct a subhamiltonian path for $H$, given $\phi$ admits a truth assignment. As the variable and connector gadgets are similar to those used for the proof of \cref{thm:st}, it suffices to show how to include the vertices of each clause gadget in the directed path $P$ that starts at $s$ and ends at $t$ passing through all the vertices of the variable gadgets. 
Consider a positive clause $c=(x\lor y\lor z)$ and let $i$, $j$, and $k$ be the clause-indices of $C_c$; the case of a negative clause is symmetric. We first show how to include in $P$ the vertices $a_c^1$ and $s_c^1$ (the vertices $a_c^2$ and $s_c^2$ can be included analogously).


\begin{itemize}
\item If $x$ is \texttt{true}, then the vertices $u^x_{6i+4}$ and $u^x_{6i+5}$ are consecutive in $P$. This allows us to extend $P$ so as to include $s_c^1$ and $a_c^1$ consecutively between $u^x_{6i+4}$ and $u^x_{6i+5}$, by adding to $P$ the  augmenting edges $u^x_{6i+4}s_c^1$ and $a_c^1 u^x_{6i+5}$; see \cref{fig:clause_gadget_poset_traversal_true}.
	
\item If $x$ is \texttt{false}, then vertices $u^x_{6i+3}$ and $u^x_{6i+4}$ are consecutive in $P$; the same holds for vertices $u^x_{6i+5}$ and $u^x_{6i+6}$. This allows us to include $s_c^1$ between $u^x_{6i+3}$ and $u^x_{6i+4}$, and $a_c^1$ between $u^x_{6i+5}$ and $u^x_{6i+6}$, by adding to $P$ the augmenting edges $u^x_{6i+3}s_c^1$, $s_c^1u^x_{6i+4}$, $u^x_{6i+5}a_c^1$, and $s_c^1u^x_{6i+4}$; see \cref{fig:clause_gadget_poset_traversal_false}. 
\end{itemize} 

\noindent Next, we show how to include in $P$ the vertex $t_c^1$ (the vertex $t_c^2$ can be included analogously).

\begin{itemize}
\item If $y$ is \texttt{true}, then the vertices $u^y_{6j}$ and $u^y_{6j+1}$ are consecutive in $P$. This allows us to extend $P$ so as to include $t_c^1$ between $u^y_{6j}$ and $u^y_{6j+1}$, by adding to $P$ the augmenting edge $t_c^1u^y_{6j+1}$.

\item If $y$ is \texttt{false}, then the vertices $u^y_{6j+1}$ and $u^y_{6j+2}$ are consecutive in $P$. This allows us to extend $P$ to include $t_c^1$ between $u^y_{6j+1}$ and $u^y_{6j+2}$, by adding to $P$ the augmenting edges $u^y_{6j+1}t_c^1$ and $t_c^1u^y_{6j+2}$.
\end{itemize}

\noindent We next focus on the remaining vertices of clause gadget $C_c$, namely $a_c$, $s_c$, $t_c$, $v_c^1$ and $v_c^2$. We distinguish three cases depending on the truth assignments for $x$ and $z$; see \cref{fig:clause_gadget_poset_xz_true,fig:clause_gadget_poset_x_true,fig:clause_gadget_poset_y_true}.

\begin{itemize}
    \item Suppose first that both $x$ and $z$ are  \texttt{true}; see \cref{fig:clause_gadget_poset_xz_true}. In this case, the vertices $u^x_{6i+2}$ and $u^x_{6i+3}$ of $L_x$ are consecutive along $P$; the same holds for the vertices $u^z_{6k+2}$ and $u^z_{6k+3}$ of $L_z$. This allows us to extend $P$ to include $s_c$, $a_c$, and $v_c^1$ consecutively between $u^x_{6i+2}$ and $u^x_{6i+3}$, by adding to $P$ the augmenting edges $u^x_{6i+2}s_c$, $a_cv_c^1$, and $v_c^1u^x_{6i+3}$, and to include $v_c^2$ and $t_c$ consecutively between $u^z_{6k+2}$ and $u^z_{6k+3}$, by adding to $P$ the augmenting edges $u^z_{6k+2}v_c^2$ and $t_cu^z_{6k+3}$. Note that the described extension of $P$ is independent of the truth assignment for $y$.

\begin{figure}[t]
    \centering
    \begin{subfigure}[b]{.49\textwidth}
    \includegraphics[page=7,scale=1.2]{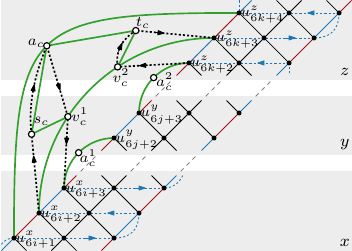}
    \caption{$x=\texttt{true}$, $z=\texttt{true}$}
    \label{fig:clause_gadget_poset_xz_true}
    \end{subfigure}
   \hfil
    \begin{subfigure}[b]{.49\textwidth}
    \includegraphics[page=8,scale=1.2]{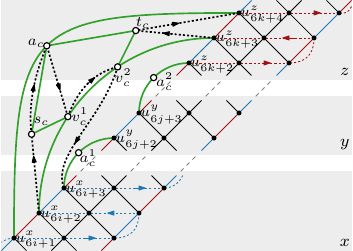}
    \caption{$x=\texttt{true}$, $z=\texttt{false}$}
    \label{fig:clause_gadget_poset_x_true}
    \end{subfigure}
    
    \begin{subfigure}[b]{.49\textwidth}
    \includegraphics[page=11,scale=1.2]{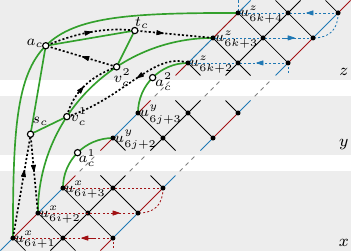}
    \caption{$x=\texttt{false}$, $z=\texttt{true}$}
    \label{fig:clause_gadget_poset_z_true}
    \end{subfigure}
    \hfil
    \begin{subfigure}[b]{.49\textwidth}
    \includegraphics[page=9,scale=1.2]{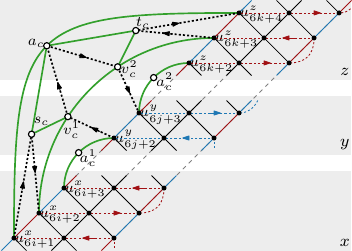}
    \caption{$x=\texttt{false}$, $z=\texttt{false}$}
    \label{fig:clause_gadget_poset_y_true}
    \end{subfigure}
    \caption{Different cases that occur while extending subhamiltonian path $P$ to include the remaining vertices of a clause gadget (i.e., those different than $a_c^1$ and $s_c^1$).}\label{fig:clause_gadget_poset_traversal}
\end{figure}

    \item Suppose next that exactly one of $x$ and $z$ is \texttt{true}. We describe the case in which $x$ is \texttt{true} and $z$ is \texttt{false} (see \cref{fig:clause_gadget_poset_x_true}); the case in which $x$ is \texttt{false} and $z$ is \texttt{true} is analogous (see \cref{fig:clause_gadget_poset_z_true}). In this case, the vertices $u^x_{6i+2}$ and $u^x_{6i+3}$ of $L_x$ are consecutive along $P$; the same holds for the vertices $u^z_{6k+3}$ and $u^z_{6k+4}$ of $L_z$. This allows us to extend $P$ to include $s_c$, $a_c$, $v_c^1$, and $v_c^2$ consecutively between $u^x_{6i+2}$ and $u^x_{6i+3}$, by adding to $P$ the augmenting edges $u^x_{6i+2}s_c$, $a_cv_c^1$, and $v_c^2u^x_{6i+3}$, and to include $t_c$ between $u^z_{6k+3}$ and $u^z_{6k+4}$, by adding to $P$ the augmenting edges $u^z_{6k+3}t_c$ and $t_c u^z_{6k+4}$. Again, the extension is independent of the truth assignment for $y$.
    \item Suppose finally that both $x$ and $z$ are \texttt{false}; see \cref{fig:clause_gadget_poset_y_true}. Since $c$ is satisfied, $y$ is \texttt{true}. In this case, the vertices $u^x_{6i+1}$ and $u^x_{6i+2}$ of $L_x$ are consecutive along $P$; the same holds for the vertices $u^y_{6j+2}$ and $u^y_{6j+3}$ of $L_y$ and for the vertices $u^z_{6k+3}$ and $u^z_{6k+4}$ of $L_z$. This allows us to extend $P$ to include $s_c$ between $u^x_{6i+1}$ and $u^x_{6i+2}$, by adding to $P$ the augmenting edges $u^x_{6i+1}s_c$ and $s_c u^x_{6i+2}$, to include $v_c^1$, $a_c$, and $v_c^2$ consecutively between $u^y_{6j+2}$ $u^y_{6j+3}$, by adding to $P$ the augmenting edges $u^y_{6j+2}v_c^1$, $v_c^1 a_c$, $a_c v_c^2$, and $v_c^2 u^y_{6j+3}$, and to finally include $t_c$ between $u^z_{6k+3}$ and $u^z_{6k+4}$, by adding to $P$ the augmenting edges $u^z_{6k+3}t_c$ and $t_c u^z_{6k+4}$.
\end{itemize}

We now prove the other direction, that is, that if $H$ is subhamiltonian, then $\phi$ is satisfiable. We start by introducing three useful properties of $H$.

\begin{prp}\label{prp:poset_upward}
The DAG $H$ has a unique upward planar embedding (up to a flip).
\end{prp}
\begin{proof}
Since $G_\phi$ is planar, we have that $H$ admits an upward planar embedding $\mathcal{E}_{H}$ with $s$ and $t$ on its outer face.
Since the underlying graph of $H$ is a subdivision of a triconnected planar graph, it has a unique combinatorial embedding. Hence, any upward planar embedding of $H$ might differ from $\mathcal{E}_{H}$ only by the choice of the outer face. The only internal faces of $\mathcal{E}_{H}$ that are incident both to a source and to a sink of $H$ are the faces of a clause gadget $C_c$ incident to $s_c$ and $t_c$, or to $s^1_c$ and $t^1_c$, or to $s^2_c$ and $t^2_c$. Suppose, for a contradiction, that an upward planar embedding $\mathcal{E}_{H}'$ of $H$ exists in which the outer face is the face of $C_c$ incident to $s_c$ and $t_c$; the argument for the other two cases is analogous. Then the outer face of $\mathcal{E}_{H}'$ is delimited by the directed paths $s_ca_ct_c$ and $s_c v^1_c v^2_c t_c$. Note that $s_c$ is only incident to the outer face of $\mathcal{E}_{H}'$ and to an internal face whose incident vertices are $s_c$, $a_c$, $u^x_{6i+1}$, $u^x_{6i+2}$, and $v^1_c$ (the face labeled $g_c$ in \cref{fig:clause_gadget_poset}). Consider the directed graph $K$ obtained from $H$ by removing the vertex $s_c$ and let $\mathcal{E}_{K}'$ be the upward planar embedding of $K$ obtained from $\mathcal{E}_{H}'$ by removing $s_c$. Then the vertices incident to the outer face of $\mathcal{E}_{K}'$ are $a_c$, $u^x_{6i+1}$, $u^x_{6i+2}$, $v^1_c$, $v^2_c$, and $t_c$. However, none of these vertices is a source of $K$, which contradicts the fact that $\mathcal{E}_{K}'$ is an upward planar embedding. 
\end{proof}


\noindent By \cref{prp:poset_upward}, in the unique upward planar embedding $\mathcal{E}_H$ of $H$ there exist several faces, formed by clause gadgets, consisting of five vertices, out of which two are sources for the face, two are sinks for the face, while the fifth one is neither a source nor a sink for the face (refer, e.g., to the shaded in gray in \cref{fig:clause_gadget_poset}). Consider such a face $f$ and denote by $v_s^1$, $v_s^2$ its two sources, and by $v_t^1$, $v_t^2$ its two sinks; see \cref{fig:clause_special_face} for an illustration.  We call $f$ a \emph{\specialsource} if the edge $v_s^1v_s^2$ or the edge $v_s^2v_s^1$ can be added inside $f$ while preserving the upward planarity of $H$ (see \cref{fig:clause_true_special_source_face,fig:clause_false_special_source_face}); otherwise we call $f$ a \emph{\specialsink} (see \cref{fig:clause_true_special_sink_face,fig:clause_false_special_sink_face}). The faces that are denoted by $g_c$, $g_c^1$ and $g_c^2$ in \cref{fig:clause_gadget_poset} are \specialsources, while the ones denoted by $h_c$, $h_c^1$ and $h_c^2$ are \specialsinks.

\begin{figure}[htb]
    \centering
    \begin{subfigure}[b]{.24\textwidth}
    \centering
    \includegraphics[page=5]{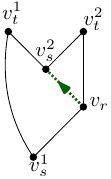}
    \caption{}
    \label{fig:clause_true_special_source_face}
    \end{subfigure}
    \hfil
    \begin{subfigure}[b]{.24\textwidth}
    \centering
    \includegraphics[page=4]{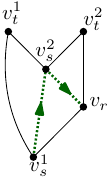}
    \caption{}
    \label{fig:clause_false_special_source_face}
    \end{subfigure}
    \hfil
    \begin{subfigure}[b]{.24\textwidth}
    \centering
    \includegraphics[page=7]{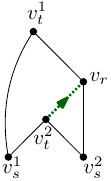}
    \caption{}
    \label{fig:clause_true_special_sink_face}
    \end{subfigure}
    \hfil
    \begin{subfigure}[b]{.24\textwidth}
    \centering
    \includegraphics[page=6]{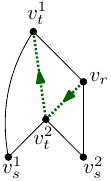}
    \caption{}
    \label{fig:clause_false_special_sink_face}
    \end{subfigure}
    \caption{Possible augmenting edges of (a-b) a \specialsource, and (c-d) a \specialsink.}\label{fig:clause_special_face}
\end{figure}

\begin{prp}\label{prp:clause_special_source}
Consider any \specialsource $f$ of $\mathcal{E}_H$ with sources $v_s^1$, $v_s^2$ and sinks $v_t^1$, $v_t^2$ in which the edge $v_s^1v_s^2$ can be added inside $f$ while preserving the upward planarity of $H$. Any subhamiltonian path $P$ for $H$ contains either (i) only the augmenting edge $v_rv_s^2$, or (ii) the two augmenting edges $v_s^1v_s^2$ and $v_s^2v_r$, where $v_r$ is the fifth vertex of $f$.
\end{prp}
\begin{proof}
Since $f$ is an internal face of $\mathcal{E}_H$, vertex $v_s^2$ can be reached only from $v_r$ or $v_s^1$. The first case (see \cref{fig:clause_true_special_source_face}) yields the augmenting edge $v_rv_s^2$ of part~(i) of the statement. In this case, no other augmenting edge can be added; indeed, the edge $v_s^1v_s^2$ would let $v_s^2$ have two incoming edges in $P$, the edge $v_rv_t^1$ would let $v_r$ have two outgoing edges in $P$, and the edge $v_t^1v_r$ would create a directed cycle $v_t^1v_rv_s^2v_t^1$. In the second case (see \cref{fig:clause_false_special_source_face}), vertices $v_s^1,v_s^2,v_r,v_t^2$ create a rhombus and, by \cref{prop:rhombus_consec}, the augmenting edge $v_s^2v_r$ or $v_rv_s^2$ must also be present, however the latter would let $v_s^2$ have two incoming edges in $P$.
\end{proof}

\noindent Symmetrically, one can prove the following.

\begin{prp}\label{prp:clause_special_sink}
Consider any \specialsink $f$ with sources $v_s^1$, $v_s^2$ and sinks $v_t^1$, $v_t^2$ of $\mathcal{E}_H$ in which the edge $v_t^2v_t^1$ can be added inside $f$ while preserving the upward planarity of $H$. Any subhamiltonian path $P$ for $H$ contains either (i) only the augmenting edge $v_t^2 v_r$, or (ii) the two augmenting edges $v_rv_t^2$ and $v_t^2v_t^1$, where $v_r$ is the fifth vertex of $f$.
\end{prp}

\begin{figure}[t]
    \centering
    \includegraphics[page=12]{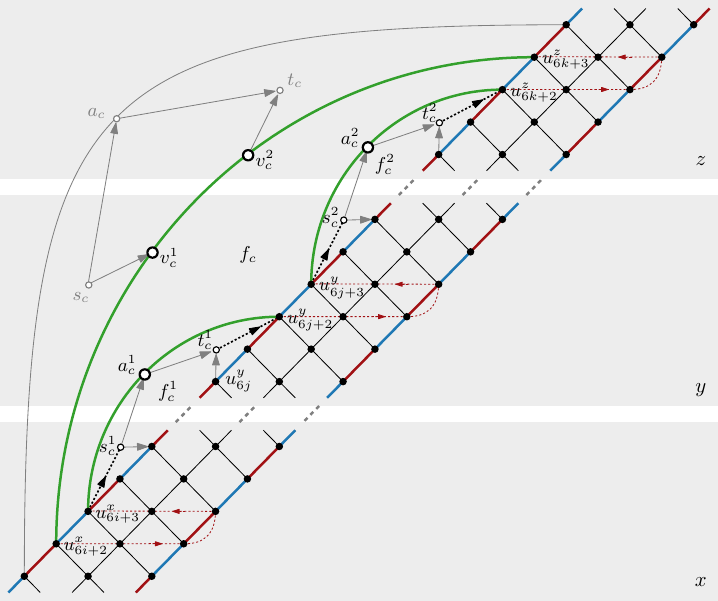}
    \caption{An unsatisfied positive clause $c=(x\lor y\lor z)$. Vertices $u^x_{6i+3}$, $u^y_{6j+2}$, $u^y_{6j+3}$ and $u^z_{6k+2}$ are incident to two augmenting edges of $P$ that are not in the interior of the anchor face $f_c$.}
    \label{fig:clause_gadget_poset_unsat}
\end{figure}

\noindent To complete the proof of the theorem, assume that there exists a subhamiltonian path $P$ for $H$, from $s$ to $t$. We compute a truth assignment as described in the proof of \cref{thm:st}. It suffices to prove that all the clauses of $\phi$ are satisfied. Let $c=(x\lor y\lor z)$ be a positive clause, where $x$, $y$ and $z$ are all \texttt{false}; the case of an unsatisfied negative clause is analogous. By \cref{prp:clause_special_source}, in the \specialsource $g_c$ there exists either the augmenting edge $u^x_{6i+2}s_c$ or the augmenting edges $u^x_{6i+1}s_c$ and $s_cu^x_{6i+2}$. As $x$ is \texttt{false}, $P$ already contains the augmenting edge $u^x_{6i+2}v^x_{6i+2}$, which implies that $u^x_{6i+1}s_c$ and $s_cu^x_{6i+2}$ are augmenting edges of $P$. Leveraging \cref{prp:clause_special_source,prp:clause_special_sink}, we can argue similarly for all other \specialsources and \specialsinks of $C_c$, as shown in \cref{fig:clause_gadget_poset_unsat}. 
In particular, $u^x_{6i+3}s_c^1$, $u^y_{6j+3}s_c^2$, $t_c^1u^y_{6j+2}$ and $t_c^2u^z_{6k+2}$ are augmenting edges of $P$. Now consider the face $f_c$, whose left boundary contains the vertices $u^x_{6i+2}$, $v_c^1$, $v_c^2$ and $u^z_{6k+3}$, and whose right boundary contains the vertices $u^x_{6i+2}$, $u^x_{6i+3}$, $a_c^1$, $u^y_{6j+2}$, $u^y_{6j+3}$, $a_c^2$, $u^z_{6k+2}$ and $u^z_{6k+3}$. Since~$f_c$ is non-transitive, there exists at least one augmenting edge of $P$, say $e$, inside $f_c$, connecting a vertex of its left boundary to a vertex of its right boundary. We already proved that edge $u^x_{6i+3}s_c^1$ is an augmenting edge of $P$ in $f_c^1$, and since $x$ is \texttt{false}, the same holds for $v^x_{6i+3}u^x_{6i+3}$. Hence, we have identified the two edges incident to $u^x_{6i+3}$ that belong to $P$, and none of them can be the edge $e$. The same holds for vertices $u^y_{6j+2}$, $u^y_{6j+3}$ and $u^z_{6k+2}$. As a consequence $e$ connects a vertex of the left boundary of $f_c$ to $a_c^1$ or $a_c^2$. Assume that $a_c^1$ is an endpoint of $e$, as the case in which $a_c^2$ is an endpoint of $e$ is analogous. As $a_c^1$ belongs to the non-transitive face $f^1_c$, by \cref{prp:non_transitive_augmenting}, there is exactly one augmenting edge of $P$ in $f_c^1$, and this edge is either $a_c^1u^x_{6i+5}$ or $u^y_{6j}a_c^1$. However, this is not possible since $x$ and $y$ are \texttt{false} and $P$ contains augmenting edges $v^x_{6i+5}u^x_{6i+5}$ and $u^y_{6j}v^y_{6j}$; a contradiction.
We conclude that $c$ is satisfied, thus completing the proof of the theorem.
\end{proof}

\section{Conclusions} \label{sec:conclusions}

In this paper, settling a long-standing conjecture of Heath and Pemmaraju~\cite{DBLP:journals/siamcomp/HeathP99} and improving upon previous results by Heath and Pemmaraju~\cite{DBLP:journals/siamcomp/HeathP99} and by Binucci et al.~\cite{binucci_et_al2019}, we have proved that deciding whether a DAG has page-number~$2$ is \NP-complete. Indeed, we have proved that the problem is \NP-hard even for $st$-planar graphs and for planar posets.

Whether our two hardness results can be combined into a single, and stronger, hardness result remains open. That is: What is the complexity of deciding whether an $st$-planar graph without transitive edges has page-number~$2$?

\bibliographystyle{splncs03}
\bibliography{references}

\end{document}